\documentclass[11pt,letterpaper]{article}
\usepackage[margin=1in]{geometry}

\usepackage{graphicx}
\usepackage{hyperref} 
\usepackage{amsmath, amsthm, amssymb}
\usepackage{tikz}
\usepackage{subcaption}
\usetikzlibrary{shapes.geometric}
\usetikzlibrary{angles,quotes}
\usetikzlibrary{calc} 
\usetikzlibrary{arrows.meta, positioning}
\usepackage{cleveref}

\usepackage{tcolorbox}
\tcbuselibrary{skins}
\usepackage{dashrule} 
\usepackage{xcolor}
\usepackage{wrapfig}
\usepackage{fullpage}

\usepackage{booktabs}
\usepackage{nicematrix}

\usepackage{todonotes}
\usepackage{array}
\usepackage{tabularx}
\usepackage{threeparttable}
\usepackage{listings}
\usepackage{float}

\lstdefinestyle{promptstyle}{
    basicstyle=\ttfamily\small,
    frame=single,
    framerule=0.4pt,
    rulecolor=\color{black!35},
    backgroundcolor=\color{black!3},
    breaklines=true,
    columns=fullflexible,
    keepspaces=true,
    xleftmargin=1em,
    xrightmargin=1em
}

\newtheorem{theorem}{Theorem}

\newtheorem{lemma}[theorem]{Lemma}
\newtheorem{proposition}[theorem]{Proposition}
\newtheorem{corollary}[theorem]{Corollary}
\newtheorem{definition}[theorem]{Definition}
\newtheorem{remark}[theorem]{Remark}

\newtheorem{assumption}[theorem]{Assumption}

\title{Polynomial-Time Mistake-Bounded Language Generation}

\author{H\'{e}ctor Jimenez\\ University of Chile\\ \texttt{hjimenez@dim.uchile.cl} \and Alexander Kozachinskiy\\National Center for Artificial Intelligence of Chile (CENIA) \\\texttt{alexander.kozachinskyi@cenia.cl} \and Vicente Opazo \\ National Center for Artificial Intelligence of Chile (CENIA) \\\texttt{vicente.opazo@cenia.cl}}

\begin{document}
\maketitle

\begin{abstract}
    In this paper, we introduce a polynomial-time version of the mistake-bounded language generation (MBLG) framework due to Kleinberg, Peale, and Reingold (2026). We obtain upper and lower bounds for several simple families of Boolean functions in our setting. 
    
    Namely, we first observe that families of parities of variables, symmetric functions, and 2CNFs are polynomial-time MBLG.
    We then show that the family of monotone functions with polynomially many maxterms is polynomial-time MBLG. For instance, disjunctions of literals are monotone Boolean functions with 1 maxterm, and thus polynomial-time MBLG. Under the existence of one-way functions, we show that disjunctions of \emph{literals} are not polynomial-time MBLG. From the latter result, we deduce that polynomial-time MBLG families are not closed under union, and that there are families that are polynomial-time PAC learnable but not polynomial-time MBLG (again, under the existence of one-way functions). Finally, assuming the existence of injective one-way functions, we show that there are polynomial-time MBLG families that are not polynomial-time PAC learnable.

\end{abstract}

\newpage

\section{Introduction}
Motivated by the remarkable empirical capabilities of LLMs, a part of theoretical research in machine learning has begun to adapt its old models to this new reality. One example of this is a recent \emph{language generation in the limit} model of Kleinberg and Mullainathan~\cite{kleinberg2024language}. It builds upon a classical language identification in the limit model due to Gold~\cite{gold1967language}, where a learner gets exposed to a stream of words of some unknown formal language from a language family $\mathcal{F}$. The task is to identify this language eventually, where the guess can be updated with each new word from the stream.  In contrast, in the Kleinberg--Mullainathan model, the task is to simply output a word from the language not seen so far in the stream. For every language in the family and for every possible enumeration of it, there has to be a moment, starting from which we get this task right. This mild change, initially motivated by the practice of using LLMs as text generators, leads to a significantly richer mathematical theory. Families, identifiable in the limit, have been characterized by Angluin~\cite{angluin1980inductive}. However, we are aware only of one ``interesting'' family satisfying these conditions -- namely, the family of pattern languages from Angluin's earlier paper~\cite{angluin1979finding}. Meanwhile, Kleinberg and Mullainathan have shown that \emph{every} countable family is generatable in the limit. This result sparked a growing interest in the model, with subsequent works establishing the result for generators, satisfying additional desirable properties like breadth~\cite{11369227,kalavasis2025limits} and representativity~\cite{pmlr-v267-peale25a}.


Although the result of Kleinberg and Mullainathan guarantees that there will be no more than finitely many mistakes, waiting for the last one can be infeasibly long.
Imagine we have just two languages in the family with a finite but very large intersection. If the stream starts with some enumeration of this intersection, we have no way to avoid a mistake when we get the last word of the intersection. Motivated by this observation, Arenas et al.~\cite{arenas2025language} show lower bounds on the size of the maximal finite intersection for several natural language families. For instance, they show there is no computable bound on the size of the intersection of two context-free languages whose intersection is finite (that is, there is no algorithm that, given two context-free languages $L_1, L_2$ such that $L_1\cap L_2$ is finite, returns some upper bound on $|L_1\cap L_2|$).

In a recent preprint, Kleinberg, Peale, and Reingold~\cite{kleinberg2026mistake} suggest that a better way of measuring success is not the time until the last mistake, but the total number of mistakes, no matter how far the last one is.  Observe that in the example, with two languages, it is easy to make at most 1 mistake. Indeed, if we see something outside the intersection, we already know the true language, and if we only see words from the intersection, we can always give another example from the intersection unless all of them are exhausted, which can happen at most once. Generalizing this example,  Kleinberg, Peale, and Reingold give a beautiful algorithm that, for finite families of size $s$, makes at most $\lceil \log_2 s \rceil$ mistakes.  Moreover, for countable families $\mathcal{F}$, they give an algorithm that, for the $i$-th language of the family, makes at most $\lceil \log_2 i\rceil$ mistakes.

\medskip

In this paper, we introduce a polynomial-time version of the mistake-bounded language generation (MBLG) framework of Kleinberg, Peale, and Reingold, in a style similar to  Valiant's polynomial-time PAC learning~\cite{kearns1987learnability} and Littlestone's online (a.k.a.~mistake-bound) learning~\cite{littlestone1988learning}. Let us first recall these models. Both start with a  $2^{poly(n)}$-size family $\mathcal{F}$ of Boolean functions $f\colon\{0, 1\}^n \to \{0, 1\}$ (or equivalently, of languages of binary words of length $n$ where these functions take value 1). In the PAC-model, the adversary secretly picks $f\in\mathcal{F}$ and a distribution $D$ over $\{0, 1\}^n$. The learner can sample pairs $(x, f(x))$ with $x\sim D$.  It gets an input $y\in\{0, 1\}^n$, sampled from $D$, with an unknown value of $f$, and its goal is to predict $f(y)$ with probability of error at most $\varepsilon$ after working $poly(n/\varepsilon)$-time.

In contrast, the online learning model is iterative:  the adversary first picks $f\in\mathcal{F}$, and then, in rounds, it picks $x\in\{0, 1\}^n$ and asks the learner --  ``what is the value of $f(x)$?''. After the learner takes its guess, the adversary reveals the correct value and proceeds to the next round. The goal of the learner is to make at most $poly(n)$ mistakes (independently of the number of rounds),  working in polynomial time (in $n$ and the number of rounds so far). 

\medskip

Similarly, in the polynomial-time MBLG that we introduce here, there is a family \[\mathcal{F} = \{L_1, \ldots, L_s\}, \qquad L_1, \ldots, L_s\subseteq\{0, 1\}^n\] of languages of binary words of length $n$, with $s = 2^{poly(n)}$. The adversary secretly picks a language $L\in\mathcal{F}$, and starts printing its words in some order (the words have to be distinct). For each $i = 0, \ldots, |L| - 1$,  seeing the first $i$
 printed words $x^1,\ldots, x^i$, our goal is to output some word $y^i\in L\setminus\{x^1, \ldots, x^i\}$. We have to make at most $poly(n)$ mistakes, working in polynomial time -- more precisely, the time to produce $y^i$ from $x^1, \ldots, x^i$ should be $poly(ni)$. 
 Families for which this is possible will be called polynomial-time MBLG families. Note that in contrast to online learning, in the mistake-bounded generation model there is no feedback -- we never get any separate confirmation whether $y^i$ was indeed from $L$.

 \medskip

 In all 3 models, if we forget about polynomial-time requirements, any $2^{poly(n)}$-size family becomes learnable. Indeed, in PAC learning, it is enough to sample $p(\varepsilon, \log_2|\mathcal{F}|)$ pairs $(x, f(x))$ for some fixed polynomial $p$, pick any $\widehat{f}\in\mathcal{F}$ that is consistent with all these pairs, and predict according to this $\widehat{f}$~\cite{valiant1984theory}. In turn, in the online learning, the ``halving algorithm'' of Littlestone will work -- predict according to the majority of functions from $\mathcal{F}$ that were never wrong so far (each mistake decreases the number of potential functions by a factor of 2, and overall we get at most $\log_2|\mathcal{F}|$ mistakes).

Note that these solutions require explicitly going through all functions in $\mathcal{F}$, and thus are not implementable in polynomial time. A similar situation is with the algorithm of Kleinberg, Peale, and Reingold~\cite{kleinberg2026mistake} for the mistake-bounded language generation. Its overall structure is as follows: it assigns weights to every function (language) of $\mathcal{F}$, with all the weights being initially equal. Each time, among the words that have not been printed so far, we choose one that maximizes the sum of weights of languages that contain it. The weights of languages are updated after each new printed example from the adversary. Once again, although the algorithm guarantees at most $\log_2|\mathcal{F}|$ mistakes, it requires explicitly going through all languages in the family, which is not implementable in polynomial time in our setting.

\section{Our Results}

\subsection{Examples}  We start with some simple upper bounds, showing that the following families are polynomial-time MBLG:
    \begin{itemize}
        \item the family of parities of variables;
        \item the family of symmetric Boolean functions;
        \item the family of 2CNFs.
    \end{itemize}
    Here, instead of families of languages, we work with families of Boolean functions, where a Boolean function $f\colon\{0, 1\}^n \to \{0, 1\}$  is identified with the language of binary words of length $n$ where $f$ takes value 1.

    \subsection{Upper Bound for Monotone Functions.}  We then provide a general upper bound for monotone Boolean functions with a few \emph{maxterms}. A maxterm of a Boolean function $f$ is a maximal input where it takes value $0$ (any further switching of an input bit from 0 to 1 leads to value 1). We show 
\begin{theorem}
\label{thm_upperbound}
    For any polynomial $p(n)$, the family of monotone Boolean functions with at most $p(n)$ maxterms is polynomial-time MBLG.
\end{theorem}

    One can notice that any monotone Boolean function, computable by a polynomial-size decision tree, has $poly(n)$ maxterms, namely, at most the number of $0$-leaves of the tree. Indeed, every maxterm ends up in some 0-leaf, so it has to be consistent with the variables and the values that are asked on the path to this leaf, and on the variables that are not asked, it has to be equal to 1 -- otherwise one can switch some input bit from  0 to 1 without changing the value. We conclude 
    \begin{corollary} For any polynomial $p(n)$, the family of monotone Boolean functions, computable by $p(n)$-size decision trees, is polynomial-time MBLG.
    \end{corollary}

    \subsection{A Lower Bound for Disjunctions of Literals.} Monotone functions with 1 maxterm are precisely disjunctions of variables. What about disjunctions of literals? We show that under a plausible cryptographic assumption -- the existence of one-way functions --  even disjunctions of \emph{three} literals cannot be learned in our framework.
\begin{theorem}
\label{thm_lowerbound}
    Assuming the existence of one-way functions, the family of disjunctions of 3 literals is not polynomial-time MBLG.
\end{theorem}
Note that the size of this family is $O(n^3)$. Hence, the algorithm of Kleinberg, Peale, and Reingold~\cite{kleinberg2026mistake} makes at most $O(\log n)$ mistakes on it. Thus, \ref{thm_lowerbound} implies that this algorithm is not implementable in polynomial time (under the existence of one-way functions). However, we stress that the actual statement is stronger -- we are claiming that no polynomial-time algorithm even with $poly(n)$ mistakes exists.

\begin{remark}
    For readers familiar with the algorithm of  Kleinberg, Peale, and Reingold, it might be instructive to understand, which part of the algorithm is not implementable in polynomial-time. As there are just $O(n^3)$ languages, it is actually not a problem to maintain the weights for each of the languages. The problem is to find a word $x$, not seen previously, maximizing the weights of languages, containing it. A naive implementation of this sub-procedure requires an exponential search, and Theorem \ref{thm_lowerbound} implies that there is probably no way of performing it in polynomial time for disjunctions of 3 literals. 
\end{remark}

Notice that disjunctions of 2 literals are special cases of 2CNFs, and thus are polynomial-time MBLG.

Our proof of Theorem \ref{thm_lowerbound} relies on a result of Rompel~\cite{10.1145/100216.100269}, who showed that one-way functions imply digital signature schemes secure against existential forgery under adaptive chosen-message attacks. More specifically, we need a version of this result, contained in Katz's book~\cite {katz2010}, where a forger cannot even obtain a new signature for a previously requested message, not just sign a new message.
Additionally, we give a simple self-contained proof of Theorem \ref{thm_lowerbound} based on the strong RSA assumption~\cite{baric1997collision,cramer2000signature,ateniese2000practical}.

\subsection{Further Observations.}  Finally, we draw some interesting conclusions from the lower bound. First, we observe that it implies that polynomial-time MBLG families are not closed under union (under the existence of one-way functions).  It also shows that MBLG can be harder than online learning and hence than PAC learning.

We also show that the opposite relationship with PAC learning is possible. Under the existence of injective one-way functions, we give an example of a polynomial-time MBLG family which is not polynomial-time PAC learnable (and thus not polynomial-time online learnable~\cite{blum1994separating}).

Finally, we show that mistake bounds of our polynomial-time algorithms -- for conjunctions, symmetric functions, parities, 2CNFs, and monotone functions with polynomially many maxterms -- are asymptotically optimal (for all algorithms, not only polynomial-time ones).




\section{Polynomial-time MBLG}

Formally, we have to work with infinite sequences  $\{\mathcal{F}_n\}_{n = 1}^\infty$ of families,where $\mathcal{F}_n = \{L_1, L_2, \ldots, L_{s(n)}\}$ and $L_1, L_2, \ldots, L_{s(n)}\subseteq \{0, 1\}^n$.

A \emph{generator} is a function $G$ that takes on input a natural number $n$ (written in unary) and a sequence of distinct binary words of length $n$, and returns a binary word of length $n$.


\begin{definition} Let $\{\mathcal{F}_n\}_{n = 1}^\infty$ be a sequence of families of languages, where $\mathcal{F}_n$ consists of languages of binary words of length $n$ for all $n$. Further, let $G$ be a generator and $e\colon\mathbb{N}\to\mathbb{N}$ be a function.

We say that $G$ makes at most $e(n)$ mistakes on $\{\mathcal{F}_n\}_{n = 1}^\infty$ if for every $n$, for every $L\in\mathcal{F}$, and for every $x^1, \ldots, x^{|L|}\in\{0, 1\}^n$ such that
\[L = \{x^1, \ldots, x^{|L|}\},\]
there exists at most $e(n)$ indices $i \in\{0, 1, \ldots, |L| - 1\}$ such that:
\[G(1^n, x^1 x^2 \ldots x^i) \notin L\setminus\{x^1, \ldots, x^i\}.\]

\end{definition}

\begin{definition}
Let $\{\mathcal{F}_n\}_{n = 1}^\infty$ be a sequence of families of languages, where $\mathcal{F}_n$ consists of languages of binary words of length $n$ for all $n$, and $|\mathcal{F}_n| \le 2^{poly(n)}$. If this sequence of families admits a polynomial-time generator, making $poly(n)$ mistakes, we say that the family is \textbf{polynomial-time MBLG}. \end{definition}

In the results below, to increase readability, we no longer explicitly consider sequences of families $\{\mathcal{F}_n\}_{n = 1}^\infty$. Instead, we simply say ``$\mathcal{F}_n$ is polynomial-time MBLG'', implicitly assuming that $n$ is a parameter, going to infinity, yielding a sequence of families.

A language $L\subseteq \{0, 1\}^n$ is identified with a Boolean function $f_L\colon\{0, 1\}^n \to \{0, 1\}$ taking value 1 precisely on words from $L$. This way we will sometimes be talking about ``families of Boolean functions'' instead of ``families of languages''. For example, a 3-CNF $\phi$ over variables $x_1, \ldots, x_n$ is identified with the language of inputs on which $\phi$ is evaluated to 1, that is, inputs that satisfy $\phi$. This allows us to consider the family of 3-CNFs over variables $x_1, \ldots, x_n$ and make statements about whether or not this family (or, more precisely, a sequence of families) is polynomial-time MBLG.

\section{Some Upper Bounds}



For a Boolean variable $x_i$ and $\alpha\in\{0, 1\}$, denote:
\[x_i^\alpha = \begin{cases}
    x_i & \alpha = 1,\\
    \lnot x_i & \alpha = 0.
\end{cases}\]
\begin{proposition}
\label{prop_conj}
The family of conjunctions of literals: 
    \[\mathcal{F} = \{ x^{\alpha_1}\land x^{\alpha_2}\land\ldots \land x^{\alpha_k}: \qquad k\ge 0,\qquad  1\le i_1 < \ldots < i_k \le n, \qquad \alpha_1\ldots \alpha_k\in\{0,1\}^k\}\] is polynomial-time MBLG with $O(n)$ mistakes. 
\end{proposition}
\begin{proof}
    Given an input sequence $(x^1,\ldots, x^i)$, we find the set $P$ of all positions where all words $x^1, \ldots, x^i$ coincide. We know $P$ includes the unknown positions $i_1, \ldots, i_k$, with the right values $\alpha_1, \ldots, \alpha_k$, but maybe some other positions too.   We try to find another word that has the same values in positions from $P$ as $(x^1,\ldots, x^i)$ -- it would be a safe guess to make. What can happen, however, is that all such words are already within $x^1, \ldots, x^i$. But then the next word of the adversary will remove some position from $P$, so this can happen at most $n$ times.
    
   Note that this strategy is implementable in polynomial time in the length of $(x^1, \ldots, x^i)$. Constructing $P$ is straightforward, and then we enumerate all words that have the required values in positions from $P$, in search of one that differs from all $x^1, \ldots, x^i$. This will take us at most $i$ steps of enumeration.

\end{proof}

\begin{proposition} 
    The family of parities of variables:
    \[\mathcal{F} = \{x_{i_1} \oplus x_{i_2} \oplus \ldots \oplus x_{i_k} \oplus c \mid 1\le k\le n,\qquad i_1 < \ldots < i_k\le n, \qquad c\in\{0, 1\}\}\] is polynomial-time MBLG with $O(n)$ mistakes.
\end{proposition}
\begin{proof}
    When we have examples $x^1, \ldots, x^i$, it is safe to output any example from their affine span over $\mathbb{F}_2$. The only way it is not possible to output a fresh string like that is when the set  $\{x^1,\ldots, x^i\}$ coincides with its own affine span. But then the next example of the adversary will increase the dimension of the span, so we will have at most $n$ mistakes.

    To implement this in polynomial time, we find a maximal independent subset of $\{x^2 - x^1, \ldots, x^i - x^1\}$, and then start enumerating all words that can be obtained as its linear combinations plus $x^1$. It takes at most $i+1$ steps to find a new word because every word out of $x^1, \ldots, x^i$ can cause us trouble just once. 
\end{proof}


\begin{proposition} The family of symmetric Boolean functions $f\colon\{0, 1\}^n\to \{0, 1\}$ is polynomial-time MBLG wth $O(n)$ mistakes.
\end{proposition}
\begin{proof}
    
    Consider $L_{f}$ the language of words that satisfy a symmetric function $f$. The algorithm receives a streaming of words $\{w_1,...,w_k\}\subseteq L_{f}$ of length $m$ and iterates over these. In the word $w_t$, for some $t \in \{1,...,k\}$, it iterates over positions searching for $1$'s. If there are a positions $j,j' \in [m]$ such that $w_t[j]=1$ and $w_t[j']=0$, it generates the word $w_t^{*} \in L_f$ by swapping bits in positions $j,j'$. If $w_t^{*} \notin \{w_1,...,w_k\}$, then the algorithm returns $w_t^{*}$. Otherwise, it continues the swapping procedure for $w_t$. The algorithm have complexity $\mathcal{O}(n^2m^2)$.

    We notice that words in $L_f$ form a graph $G=(V,E)$, where $V=L_f$ and edges are defined by the relation $ E:=\{ (u,v) \mid u \text{ can generate } v \text{ by swapping a bit 1 with a bit 0} \} $. Moreover, all the words that share the same number of ones form a connected component. Therefore, we can always generate a new edge ($w_t^{*}$) of the graph $G$, except in the cases when the adversary sends all vertices (words in the stream) in a connected component. In those cases, the algorithm could make a mistake. But we notice that this happens at most $n$ times (there are at most $n$ connected components), so the algorithm can have at most $n$ mistakes. 
    
    
    
\end{proof}

The next example generalizes the family of conjunctions of variables.

\begin{proposition} 
\label{prop_2cnf}
The family of 2-CNFs over variables $x_1, \ldots, x_n$ is polynomial-time MBLG with $O(n^2)$ mistakes.
\end{proposition}
\begin{proof}
Given the adversary's examples $w^1, \ldots, w^i\in\{0,1\}^n$, we calculate all disjunctions of 2 literals that take value 1 on all these examples, and compose them into a 2-CNF $\phi_i$. All clauses of the adversary's real 2-CNF $\phi$ are included in the clauses of $\phi_i$. Hence, it is a ``safe guess'' to output any satisfying assignment to $\phi_i$ different from $w^1,\ldots, w^i$. If no such assignment exists, we might make a mistake. But then the next adversary's example will not satisfy $\phi_i$, and the number of clauses in $\phi_{i+1}$ will be strictly smaller. Hence, we will make at most $O(n^2)$ mistakes.

It only remains to explain how to find a satisfying assignment $w$ to a 2-CNF $\phi_i$ such that $w\neq w^1, \ldots, w\neq w^i$,    in polynomial time (if such $w$ exists). For that, we go over all $w^\ell = w^\ell_1\ldots w^\ell_n$, $\ell = 1, \ldots, i$, Next, we go through all prefixes of $w^\ell$, inverting their last bit. Whenever we obtain a word which is not a prefix of any of the words $w^1,\ldots,w^i$, we try to find a satisfying assignment to $\phi_i$, consistent with this word (some number of the initial variables are fixed as in this word) -- by simplifying $\phi_i$ and then using the polynomial-time algorithm for 2-SAT. 

Whenever we obtain a satisfiable 2-CNF in this way, its solution will be a new satisfying assignment to $\phi_i$. On the other hand, if $w$ is a satisfying assignment to $\phi_i$, distinct from $w^1, \ldots, w^i$, it will be considered at some point in our algorithm. Indeed, consider the smallest prefix of $w$ which is not a prefix of any $w^1, \ldots, w^i$. A prefix which is one bit shorter will be a prefix of some $w^\ell$, and hence $w$ will be a solution to a simplification of $\phi_i$, arising by inverting the bit after this prefix.

\end{proof}

\section{Monotone Functions with Polynomially Many Maxterms}

We consider the ``inclusion order'' on $\{0, 1\}^n$:
\[x\preceq y \iff \forall i = 1, \ldots, n\,\, x_i \le y_i \qquad x,y\in\{0, 1\}^n.\]

A function $f\colon\{0, 1\}^n$ is monotone if $x\preceq y \implies f(x) \le f(y)$ for all $x, y\in\{0, 1\}^n$. A maxterm of $f$ is a maximal element of $f^{-1}(0)$.
\begin{theorem}[Restatement of Theorem \ref{thm_upperbound}]
For any polynomial $p(n)$, the family of monotone Boolean functions with at most $p(n)$ maxterms is polynomial-time MBLG with at most $O(p(n)\cdot n)$ mistakes.
\end{theorem}
\begin{proof}
     Imagine that words, printed by the adversary, get marked. Each time, we have to output a non-marked word. Since we consider only monotone functions, it is safe to output any non-marked word if below it (in the inclusion order) there is some marked word. So we only have to care about ``crucial'' moments when the set of marked words is upwards-closed. 

    In these moments, we will always output some maximal non-marked word. We make a mistake when we output a word from $f^{-1}(0)$. Thus, we can only make a mistake when we output a maxterm of $f$.  To make sure that we make at most $poly(n)$ mistakes, we will give an algorithm such that every word $x\in\{0, 1\}^n$ gets outputted at most $O(n)$ times at crucial moments.

    We will achieve this via the following simple strategy. We maintain the set  $M$ of maximal non-marked elements. Note that the size of $M$ is polynomial in the bit-length of $(x^1, \ldots, x^i)$ because every its element has some string from $\{x^1, \ldots, x^i\}$ directly above it, but every string $\{x^1, \ldots, x^i\}$ can have at most $n$ elements of $M$ directly below it. Further, we maintain how many times each element of $M$ has been given to the output at crucial moments. When there is a crucial moment, we choose any element of $M$ that has been given the minimum number of times.

    Let us study how the set $M$ evolves at crucial moments. The key observation is that each time, at least one element of $M$ goes away (and maybe some new elements are added). Moreover, elements that go away never come back. Indeed, in the crucial moment, all strings under any element of $M$ are non-marked (because marked strings are upwards-closed). Hence, if all elements of $M$ had survived, the set of non-marked elements would increase or stay the same. However, this set decreases over time. Moreover, an element $m\in M$ can go away only if $m$ or something below it gets a mark, and then it can never enter $M$ again.

Since strings never return to $M$, and since the size of $M$ is polynomial, maintaining $M$ together with the number of times its elements have been given as an output takes polynomial time. 

Next, these numbers of times -- how many times each element of $M$ has been used at  crucial moments -- satisfy the rules of the following game with numbers written on a board:

\begin{itemize}
    \item initially, one can write any number of 0's on the board (this corresponds to the first crucial moment when the set $M$ is created, and no string has been given as the output yet).
    \item Then you work in steps. At each step, you first increase the minimal number on the board by 1 (if there are multiple minimal numbers, you increase just one). Then at least one number has to be erased from the board, and then some new zeroes can be added.
\end{itemize}

In our application, the number of steps is bounded by the number of crucial moments, which is at most $2^n$. It now suffices to establish the following lemma.
\begin{lemma} 
Let $n\ge 1$.  Then the minimal number of steps for a number $n$ to appear on the board in the game with numbers is at least $2^{n - 1}$.
    \end{lemma}
    \begin{proof}
    Imagine numbers as stacks of coins. Somewhat strangely, a number $m$ will be represented as a stack of $m + 1$ coins -- a level-$0$ coin, a level-1 coin, ..., a level-$m$ coin. Coins at different levels will have different costs. Namely, a level-$0$ coin has cost 1, and  a level-$i$ has cost $2^{i-1}$ for $i\ge 1$.

    Look at the total cost of all the coins. When we add 0, the total cost increases by 1. Now, when we add +1 to a number $m$, we create a level-$(m+1)$ coin whose cost is $2^m$. But then we either delete it right away, or we delete some other number that has to be at least $m$.  The cost of all the coins that we delete is at least $1 + 1 + 2 + \ldots + 2^{m-1} = 2^m$, so the total cost does not increase.

    Hence, to create a coin at level $n$, we have to put at least $2^{n-1}$ 0-level coins. Although we can put a lot of 0-level coins at a time, each 0-level coin can be converted into a 1-level coin only one at a time, so we require at least $2^{n-1}$ steps.
    \end{proof}

\section{A Lower Bound for Disjunctions of 3 Literals}
In this section, we establish Theorem \ref{thm_lowerbound}. The first step
is a reduction from 3CNFs.
\begin{lemma}
    If the family of disjunctions of 3 literals is polynomial-time MBLG, then the family of 3-CNFs is polynomial-time MBLG.
\end{lemma}
\begin{proof}
    Apply the algorithm for disjunctions of 3 literals, without changes, to 3-CNFs. Each 3-CNF $\phi$ is written as
    \[\phi = D_1 \land D_2 \ldots \land D_m,\]
    where $m \le O(n^3)$ and $D_1, \ldots, D_m$ are disjunctions of 3 literals. Every example we see that satisfies $\phi$ satisfies each of the disjunctions $D_1, \ldots, D_m$. Hence, for each $i = 1, \ldots, m$, our algorithm for disjunctions outputs a word not satisfying $D_i$ at most $p(n)$ times for some polynomial $p(n)$. Hence, it outputs a word, not satisfying $\phi$, at most $mp(n) = O(n^3 p(n))$ times.
\end{proof}

The next step is to reduce from polynomial-size circuits. That is, we show that if 3CNFs were polynomial-time MBLG, then so would be polynomial-size circuits, but in the \emph{open regime}. The latter means a modification of the MBLG setting, where the adversary in the beginning \emph{reveals} to us the language it gives examples from (in our case, a polynomial-size circuit). In this regime, the complexity of identifying the language completely goes away, and we are only left with the problem of generating new examples, satisfying a given circuit.

\begin{lemma}
If the family of 3CNFs is polynomial-time MBLG, then  for any polynomial $p(n)$, the family of $p(n)$-size circuits is polynomial-time MBLG in the open regime.    
\end{lemma} 
\begin{proof}
    In the open MBLG for $p(n)$-size circuits, we are seeing a circuit $C$ and some strings $x^1, \ldots, x^m$ that satisfy it. Our task is to come up with a new example. Take a standard reduction from Circuit-SAT to 3-SAT where a circuit $C$ is transformed into polynomial-size 3CNF $\phi_C$ (over variables of $C$ and some fresh variables) such that satisfying assignments to $C$ are in one-to-one correspondence with satisfying assignments to $\phi_C$ (and are easily computable from one another). Use the polynomial-time algorithm for 3CNFs, converting adversary examples of satisfying assignments to $C$ into satisfying assignments to $\phi_C$, and then new examples of satisfying assignments to $\phi_C$ from the algorithm into satisfying assignments to $C$.
\end{proof}

The proof is concluded by the following lemma.
\begin{lemma}
\label{lem_reduction_OW}
If for any polynomial $p(n)$, the family of $p(n)$-size circuits is polynomial-time MBLG in the open regime, then one-way functions do not exist.
\end{lemma}
\begin{proof}
    We show that if for any polynomial $p(n)$, the family of $p(n)$-size circuits is polynomial-time MBLG, then one can forge in polynomial time any digital signature scheme via an adaptive chosen-message attack. A digital signature scheme consists of 
    \begin{itemize}
        \item A polynomial-time randomized algorithm $Gen$ that, given $1^k$, where $k$ is a security parameter, produces a pair of strings $(sk, pk)$, where $sk,pk$ are called the secrete and the public key, respectively; the value of $k$ is easily computable from $pk$;
        
\item a deterministic polynomial-time algorithm $Vrf$ that takes on input a public key $pk$, a ``message'' $m\in\{0, 1\}^{q(k)}$ and a ``signature'' $\sigma\in\{0, 1\}^{r(k)}$, for some polynomials $q(k), r(k)$. It returns either 0 or 1;

\item a deterministic polynomial-time algorithm $Sign$ that takes on input $pk, sk$ and $m\in\{0, 1\}^{q(k)}$ and outputs a binary strings $\sigma\in\{0, 1\}^{r(k)}$ such that $Vrf(pk, m, \sigma) =1$.
             \end{itemize}
An strong chosen-message forger is a polynomial-time algorithm that gets $pk$ on input and has oracle success to $Sign(pk, sk, \cdot)$ -- that is,  for polynomially many messages $m_1, \ldots, m_\ell\in\{0, 1\}^{q(k)}$ of its choice, it can get the values 
\[\sigma_1 = Sign(pk, sk, m_1),\,\, \sigma_2 = Sign(pk, sk, m_2), \ldots,\,\,  \sigma_\ell = Sign(pk, sk, m_\ell)\]
without knowing $sk$. The value of $m_i$ might depend on $\sigma_1, \ldots, \sigma_{i - 1}$. The algorithm has to output pair of strings $m\in \{0, 1\}^{q(k)}$ and $\sigma \in\{0, 1\}^{r(k)}$ such that $Vrf(pk, m, \sigma) = 1$ and
\[(m, \sigma)\notin \{(m_1, \sigma_1), \ldots, (m_\ell, \sigma_\ell)\}.\]
Note that the forger is allowed to get a new valid signature for one of the messages $m_1, \ldots, m_\ell$ as well as sign a new message.

 We give a strong chosen-message forger for any digital scheme, assuming that the family of $p(n)$-size circuits admits a polynomial-time MBLG $A$ in the open regime, for any polynomial $p$. Given a public key $pk$, the forger constructs a polynomial-size circuit computing $Vrf(pk, \cdot, \cdot)$ where the second and the third inputs are of lengths $q(k)$ and $r(k)$, respectively. The forger simulates the language generation game using this circuit, playing as the adversary against the MBLG-algorithm $A$ (for circuits of size at most the size of $Vrf(pk, \cdot, \cdot)$). Namely, it simply enumerates messages $m_1, m_2, m_3,\ldots$, getting satisfying assignments to $Vrf(pk, \cdot, \cdot)$ by requesting signatures from the sign-oracle. After at most polynomially many steps, since $A$ makes at most polynomially many mistakes, $A$ has to provide a new signature for a new message, or a new message with a signature, forming a new satisfying assignment to $Vrf(pk, \cdot, \cdot)$.

 It remains to refer to Theorem 3.5 of Katz~\cite{katz2010}, stating that under the existence of one-way functions, there exists a digital signature scheme for which there is no strong chosen-message forger (``strongly unforgeable under an adaptive chosen-message attack'').

\end{proof}

\section{A Proof Based on the Strong RSA Assumption}
Here we give a self-contained proof of Lemma \ref{lem_reduction_OW} under the strong RSA assumption whose exact statement is as follows.
\begin{assumption}[\cite{baric1997collision}]
\label{as_rsa}
    There is a polynomial-time randomized algorithm $G$, on input $1^n$, generating two distinct $n$-bit primes $q_1,q_2$, such that the following holds.

    There is no probabilistic polynomial-time algorithm that, on input
    \[(1^n, N, x)\]
    where $N = q_1q_2$ with $(q_1, q_2)\sim G(1^n)$ and $x\in_R\mathbb{Z}_N^*$, finds $y\in\mathbb{Z}_N^*$ and a prime $e < N$ such that $y^e \equiv x \pmod{N}$ with probability at least $1/poly(n)$.
    
\end{assumption}
\begin{lemma}
\label{lem_reduction}
If for any polynomial $p(n)$, the family of $p(n)$-size circuits is polynomial-time MBLG in the open regime, then the strong RSA assumption fails.
\end{lemma}
\begin{proof}
Consider a polynomial-size circuit $C$ that, given
\[N\in\mathbb{N} \text { with $2n$ bits}, \,\, x\in\mathbb{Z}_N^*, \,\ y\in\mathbb{Z}_N^*, \,\, e\in\{0,1,\ldots, N-1\},\]
checks, whether $e$ is a prime and that $y^e\equiv x \pmod{N}$ (it is well-known that primality testing and exponentiation in $\mathbb{Z}_N^*$ can be done in polynomial time). For any fixation of $N,x$, we obtain a circuit $C_{N,x}$ of size at most $size(C)+ O(n)$, not exceeding some fixed polynomial in $n$. Take a polynomial-time MBLG-algorithm $A$ for circuits of such size in the open regime, and let $p(n) = n^{O(1)}$ be its mistake bound.

We will use this MBLG-algorithm to refute the strong RSA assumption. That is, given $(1^n, N, x)$ as in Assumption \ref{as_rsa}, with $N = q_1q_2$ for two distinct $n$-bit primes $(q_1, q_2)\sim G(1^n)$, and for $x\in_R\mathbb{Z}_N^*$, with probability $0.9$, in polynomial time, we will find $y\in\mathbb{Z}_N^*$ and a prime $e <N$
 such that $y^e \equiv x\pmod{N}$.

We sample $t = np(n)$ random primes $p_1, \ldots, p_t$ less than $N$. We simulate the adversary in the game against the MBLG algorithm $A$, by playing the circuit:
\[C_{N, \widehat{x}}, \qquad \widehat{x} = x^{p_1 p_2 \ldots p_t}\pmod{N}.\]
Satisfying assignments to $C_{N, \widehat{x}}$ are pairs $(e, y)$ where $e < N$ is a prime and $y\in\mathbb{Z}_N^*$ is such that $y^e \equiv \widehat{x}\pmod{N}$. We continue the simulation of the adversary by playing the following satisfying assignments to $C_{N,\widehat{x}}$:
\begin{equation}
    \label{eq_simulation}
    (p_1, x^{p_2\ldots p_t}), (p_2, x^{p_1p_3\ldots p_t}), \ldots (p_t, x^{p_1\ldots p_{t -1}})
\end{equation}
(all exponents are taken modulo $N$).

We will succeed if the MBLG algorithm ever outputs a satisfying assignment $(q, z)$ to $C_{N, \widehat{x}}$ with $q\neq p_1, \ldots, q\neq p_t$. Indeed,  then we can find $y\in\mathbb{Z}_N^*$ with $y^q \equiv x\pmod{N}$ via the following trick due to Shamir~\cite{shamir1983generation}. Namely, by definition, we have:
\[z^q \equiv \widehat{x} = x^{p_1\ldots p_t} \pmod{N}.\]
Then, since $p_1\ldots p_t$ and $q$ are coprime, we can use the Euclid's algorithm to find $\alpha, \beta\in\mathbb{Z}$ with 
\[\alpha p_1\ldots p_t + \beta q = 1.\]
Set
\[y = z^\alpha \cdot x^{\beta}.\]
Observe that
\begin{align*}
    y^q = z^{\alpha q} \cdot x^{\beta q} = (z^q)^{\alpha} \cdot x^{\beta q} =  x^{\alpha p_1\ldots p_t + \beta q} = x\pmod{N}.
\end{align*}

We now show that we will succeed with high probability. It is enough to lower bound the success probability conditioned on the event $E$ that  $p_1, \ldots, p_t$ are all distinct and are coprime to $\phi(N)$ (the value of Euler's totient function on $N$). Indeed, this event has probability exponentially close to 1, since we are sampling polynomially many primes out of $N/\ln N \sim 2^{2n}/n$, and since there are $O(\ln N) = O(n)$ distinct prime divisors of $\phi(N)$.

Observe first that conditioned on $E$, the simulation of the adversary in \eqref{eq_simulation} is correct in the sense that it gives distinct examples.

Next, conditioned on $E$, the distribution of $p_1, \ldots, p_t$ is a uniform drawing without replacement from the set $\mathcal{P}$  of primes smaller than $N$ that are not divisors of $\phi(N)$. 
The key observation is that, conditioned on $E$, for $x\in_R\mathbb{Z}_N^*$ independent of $p_1, \ldots, p_t$, the distributions of the tuples 
\[(x, p_1, \ldots, p_t)\]
and
\[(x^{p_1\ldots p_t}\mathrm{mod }N, p_1, \ldots, p_t)\]
coincide. This is due to a fact that for an exponent $p$, coprime with $\phi(N)$, the mapping $x\mapsto x^p$ is a permutation of $\mathbb{Z}_{N}^*$ \footnote{Indeed, let $p'$ be the inverse of $p$ modulo $\phi(N)$, and observe that $(x^{p'})^{p}\equiv x \pmod{N})$} . Hence, for any fixed values of $p_1, \ldots, p_t\in\mathcal{P}$, the distribution of $x^{p_1\ldots p_t} \mathrm{mod }N$ is uniform over $\mathbb{Z}_N^*$ which proves the equality of two distributions. 

Thus, to analyze the probability of success, we can assume that the MBLG algorithm is given $x$ instead of $\widehat{x}$, and then fed
\[(p_1, x^{1/p_1}), (p_2, x^{1/p_2}), (p_3, x^{1/p_3}), \ldots \]
where $p_1, \ldots, p_t$ are drawn without replacement from $\mathcal{P}$ independently of $x$ (here $x^{1/p}$ denotes $y\in\mathbb{Z}_N^*$ such that $y^p \equiv x \pmod{N}$ which is unique for $p\in\mathcal{P}$). For at least one $i = 1, \ldots, n$, since $t > p(n)$, after seeing the first $i$ examples, the algorithm will output some  $(q_i, y_i)$  such that $y_i^{q_i}\equiv x\pmod{N}$ and
\[(q_i, y_i) \neq (p_1, x^{1/p_1}), \ldots, (q_i, y_i)\neq (p_i, x^{1/p_i}).\]
In particular, $q_i$ is distinct from $p_1, \ldots, p_i$ (again, because for $p  \in\mathcal{P}$, there is just a single possible value of $x^{1/p}$). Likewise, with high probability, $q_i$ is distinct from $p_{i+1}, \ldots, p_t$. Indeed, $q_i$ is a function of $x, p_1, \ldots, p_i$, and for each of their possible values, $p_{i+1}, \ldots, p_t$ are sampled uniformly at random from an exponentially large set. The probability that one of these polynomially many numbers will coincide with $q_i$ is exponentially small.

\end{proof}

\section{Further observations}

\begin{proposition}
Under the existence of one-way functions, polynomial-time MBLG families are not closed under union.
\end{proposition}
\begin{proof}
    Consider the families $\mathcal{F}_{pos}, \mathcal{F}_{neg}$ of disjunctions of 3 literals with, respectively, at most 1 positive and at most 1 negative literal. The family $\mathcal{F}_{pos}$ is polynomial-time MBLG by the same argument as in  Proposition \ref{prop_2cnf} because Horn-SAT (the CNF-SAT problem where each clause has at most one positive literal) is polynomial-time solvable~\cite{DOWLING1984267}. Likewise, $\mathcal{F}_{neg}$ is polynomial-time MBLG by a dual argument. But the union of these 2 families gives the set of all disjunctions of 3 literals, which is not polynomial-time MBLG under the existence of one-way functions by Theorem \ref{thm_lowerbound}.
\end{proof}

We conclude by comparing polynomial-time MBLG with polynomial-time PAC learning. 

\begin{definition}
     Let $\{\mathcal{F}_n\}_{n = 1}^\infty$ be a sequence of families of Boolean functions, where $\mathcal{F}_n$ consists of Boolean function with input length $n$. Assume that $|\mathcal{F}_n| \le 2^{poly(n)}$. We say that  $\{\mathcal{F}_n\}_{n = 1}^\infty$ is \textbf{polynomial-time PAC-learnable} if there exists an algorithm with the following properties:
     \begin{itemize}
         \item on input, it gets $1^n$, a rational number $\varepsilon$, a word $\widehat{x}\in\{0, 1\}^n$, and sample access to a pair $(f, D)$, where $f\in \mathcal{F}_n$ and $D$ is a probability distribution over $\{0, 1\}^n$ (meaning that it can sample pairs $(x, f(x))$ for $x\sim D$, but does not have an explicit description of these objects);
         \item it works in $poly(n, 1/\varepsilon)$ time;
         \item for any $(f, D)$ with $f\in\mathcal{F}_n$, with probability $1 - \varepsilon$ over $\widehat{x}\sim D$ (and over the internal randomness of the algorithm, including sampling $(f,D)$), the output of the algorithm coincides with $f(\widehat{x})$.
     \end{itemize}

\end{definition}

For the family of disjunctions of 3 literals, whose size is $O(n^3)$, the halving algorithm of Littlestone~\cite{littlestone1988learning} is polynomial-time implementable, which means that this family is polynomial-time online-learnable, and hence polynomial-time PAC learnable~\cite{93335.93365}. Hence, by Theorem \ref{thm_lowerbound}, we obtain:

\begin{corollary}
    Under the existence of one-way functions, there exists a family which is polynomial-time PAC learnable but not polynomial-time MBLG.
\end{corollary} 

As we show, assuming the existence of injective one-way functions, the opposite behaviour is possible. We will now need to recall a formal definition of one-way functions.

\begin{definition}
\label{def_oneway}
    A polynomial-time computable function $f\colon\{0, 1\}^* \to \{0, 1\}^*$ is called one-way if there is no polynomial-time randomized algorithm that for infinitely many $n$, with probability $1/poly(n)$, given $1^n$ and $f(x)$ with $x\in_R\{0, 1\}^n$, returns some $\widehat{x}\in\{0, 1\}^n$ such that $f(x) = f(\widehat{x})$.
\end{definition}

Note that for a general one-way function, it should not be possible to return any element of the pre-image of $f(x)$. For injective one-way functions, this preimage consists only of $x$ itself.

\begin{proposition}
    Assuming that there exists an injective one-way function, there exists a polynomial-time MBLG family which is not polynomial-time PAC learnable.
\end{proposition}
\begin{proof}
    Let $f$ be an injective one-way function. Let $f_n$ denote its restriction to $n$-bit words. Without loss of generality, we may assume that its outputs all have length $p(n)$ for some polynomial $p$. Our family will have just one language $L$, consisting of all pairs $(i, y) \in \{1, \ldots, n\}\times \{0, 1\}^{p(n)}$ such that there exist $x\in\{0, 1\}^n$ with $f(x) = y$ and $x_i = 1$.

    On the one hand, this singleton family is polynomial-time MBLG. Indeed, seeing
    \begin{equation}
        \label{eq_oneway}
        (i^1, y^1), \ldots, (i^m, y^m)\in L
    \end{equation}
    from the adversary,
    we start going through all indices $i =1,\ldots, n$ and all $n$-bit words $x$ with $x_i = 1$, computing $y = f(x)$, until $(i, y)$ is not in \eqref{eq_oneway}. Since $f$ is injective, we will require at most $m + 1$ such steps unless the whole $L$ is already listed. Thus, it takes us $poly(mn)$ time.

On the other hand, we show that a polynomial-time PAC learner for this family can be turned into a polynomial-time inverter of $f$. Indeed, denote by $\phi$ the function, taking value 1 exactly on words from $L$. For $i\in\{1,2,\ldots, n\}$, let $D_i$ be a distribution of $(i, f(x))$, where $x\in_R\{0, 1\}^n$. The sample access to $(\phi, D_i)$ can be simulated by the inverter itself. Namely, we can sample $x\in_R\{0, 1\}^n$ and compute $(i, f(x))$. By injectivity of $f$, the value of $\phi$ on $(i, f(x))$ is one if and only if $x_i = 1$. Using the PAC-learner for $\varepsilon = 1/n^2$, for every $i = 1, \ldots, n$, we get a polynomial-time algorithm that, with probability $1 - 1/n^2$, given $\widehat{y} = f(\widehat{x})$ for $\widehat{x}\in_R\{0, 1\}^n$, returns the $i$-th bit of $\widehat{x}$. Using it for all $i$, we obtain a contradiction with the fact that $f$ is one-way.

\end{proof}


\section{Optimality of Mistake Bounds}

We now establish that all our previous algorithms have optimal mistake bounds, even if one takes into consideration not only algorithms that run in polynomial time. For that, we give a general mistake bound for families that we call \emph{$(p, Q)$-rich}, as defined below.

Let us call a family of sets $B_1, \ldots, B_Q\subseteq\{0, 1\}^n$ \emph{pierceable} if for every $i = 1, \ldots, Q$, there exists a word that belongs to all sets $B_1, \ldots, B_Q$ except $B_i$.
Now, a family $\mathcal{F}$ of subsets of $\{0, 1\}^n$ is called $(p, Q)$-rich if there exists a pierceable
 family of sets $B_1, \ldots, B_Q\subseteq\{0, 1\}^n$ such that $\mathcal{F}$ contains all possible intersections of $p$ of these sets.

\begin{proposition}
\label{prop_generallb}
    Let $p, Q\in\mathbb{N}, p\le Q$. Then for any $(p, Q)$-rich family $\mathcal{F}$, any language generation algorithm makes at least $\Omega(p\ln(Q/p))$ mistakes on $\mathcal{F}$.
\end{proposition}
\begin{proof}
Let $B_1, \ldots, B_Q\subseteq \{0, 1\}^n$ be a pearcable family of sets such that any intersection of $p$ of these sets belongs to $\mathcal{F}$.

    Consider the following randomized strategy of the adversary. It picks a uniformly random permutation $\pi$ of $\{1, \ldots, Q\}$. Then it enumerates all words in the intersection
    \[B_{\pi(1)} \cap B_{\pi(2)} \cap \ldots \cap B_{\pi(Q)}.\]
    Once it is done, it names a word that belongs to all sets $B_1, \ldots, B_Q$ except $B_{\pi(1)}$, exploiting the fact that $B_1, \ldots, B_Q$ is a piearcable family.
In particular, this word is still in 
\[B_{\pi(2)} \cap B_{\pi(3)} \cap \ldots \cap B_{\pi(Q)}.\]
After that, the adversary enumerates all words of $B_{\pi(2)} \cap B_{\pi(3)} \cap \ldots \cap B_{\pi(Q)}$ that have not been listed by it before. Once it is done, it names a word that belongs to all sets $B_1, \ldots, B_Q$ except $B_{\pi(2)}$. This word is still in the intersection  $B_{\pi(3)} \cap \ldots \cap B_{\pi(Q)}$. The adversary continues working in the same manner with the intersection   $B_{\pi(3)} \cap \ldots \cap B_{\pi(Q)}$, and so on. It stops at the moment when all words of the intersection $B_{\pi(Q-p+1)} \cap \ldots \cap B_{\pi(Q)}$ of the last $p$ sets are listed.  This intersection belongs to $\mathcal{F}$ and thus the adversary produces a valid play in the language generation games.

Fix any language generation algorithm. We give a lower bound on its expected number of mistakes when played against the above strategy of the adversary. The adversary's enumeration can be separated into $Q - p + 1$ stages, where during the $i$-th stage, for $i = 1, \ldots, Q - p + 1$, it enumerates all the words of $B_{\pi(i)}\cap \ldots \cap B_{\pi(Q)}$ that have not been enumerated before. By the end of the $i$-th stage, all words of $B_{\pi(i)}\cap \ldots \cap B_{\pi(Q)}$ are already enumerated by the adversary. At this moment, for $i = 1, \ldots, Q - p$, the language generation algorithm is thus obliged to name a word $w_i \notin B_{\pi(i)}\cap \ldots \cap B_{\pi(Q)}$. There will be an index $k(i)\in\{\pi(i), \ldots, \pi(Q)\}$ such that $w_i\notin B_{k(i)}$.  The algorithm will make a mistake if $k(i)$ coincides with one of the indices $\pi(Q - p + 1), \ldots, \pi(Q)$.

Now, observe that for a fixed language generation algorithm, $k(i)$ is a function of the adversary's enumeration in the first $i$ stages. In these stages, we are enumerating intersections from $B_{\pi(1)} \cap B_{\pi(2)} \cap \ldots \cap B_{\pi(Q)}$ up to $B_{\pi(i)} \cap B_{\pi(2)} \cap \ldots \cap B_{\pi(Q)}$. The intersection of sets is a commutative operation, and thus to compute these intersections, we just need to know the values $\pi(1), \ldots, \pi(i - 1)$. In other words, the play up to the end of the $i$-th stage is determined by $\pi(1), \ldots, \pi(i - 1)$, and thus $k(i)$ is a function of $\pi(1), \ldots, \pi(i - 1)$. 

Now, for any fixation of $\pi(1), \ldots, \pi(i - 1)$, the distribution of each of $\pi(i), \ldots, \pi(Q)$ is uniform in $\{\pi(i), \ldots, \pi(Q)\}$. Thus, each of these numbers has probability $1/(Q - i + 1)$ to coincide with $k(i)$, and these events are disjoint for different indices from $i$ to $Q$. In particular, with probability $p/(Q - i + 1)$, one of the numbers $\pi(Q - p + 1), \ldots, \pi(Q)$ will coincide with $k(i)$, causing a mistake.

Summing over $i = 1, \ldots, Q - p$ the expected number of mistakes is thus at least:
\[\frac{p}{Q} + \frac{p}{Q-1} + \ldots + \frac{p}{p + 1} = \Omega(p(\ln(Q/p))),\]
as required.

\end{proof}
\begin{corollary} The minimal number of mistakes in a language generation algorithm is
    \begin{itemize}
        \item[a)] $\Omega(n)$ for the family of conjunctions of variables;
        \item[b)] $\Omega(n \cdot p(n))$, for any polynomial $p(n)$ and for the family of monotone Boolean functions with at most $p(n)$ maxterms;
        \item[c)] $\Omega(n)$ for the family of symmetric functions;
        \item[d)] $\Omega(n^2)$ for the family of 2CNFs;

        \item[e)] $\Omega(n)$ for the family of parities of variables.
    \end{itemize}
\end{corollary}
\begin{proof}
    \begin{itemize}
        \item[a)] The family of conjunctions of variables is $(n/2, n)$-rich and thus Proposition \ref{prop_generallb} gives a lower bound $\Omega(n\ln(n/(n/2)) = \Omega(n)$. Indeed, for $i = 1, \ldots, n$, let $B_i$ be the set of inputs with the $i$-th coordinate equal to 1. The family $B_1, \ldots, B_n$ is pierceable because for every $i= 1, \ldots, n$, the word with ones every except the $i$-th coordinate belongs to all sets $B_1, \ldots, B_n$ except $B_i$. In turn, the intersection of any $n/2$ of these sets is realizable as a set of 1s of a conjunction of $n/2$ elements.

        \item[b)] This family is $(p(n), 2^{\Omega(n)})$-rich, which gives us a desired lower bound due to Proposition \ref{prop_generallb}.  Namely, let $\binom{[n]}{n/2}$ denote the set of $n/2$-size subsets of $[n] = \{1, \ldots, n\}$. For $S\subseteq [n]$, let $B_S\subseteq \{0, 1\}^n$ denote the set of $n$-bit binary strings that have at least one 1 in coordinates from $S$. Observe that the family $\{B_S \mid S\in \binom{[n]}{n/2}\}$ is piearceable. Indeed, for any $S\in \binom{[n]}{n/2}$, the string that has 1s in the positions of the complement to $S$ does not belong to $B_S$ but belongs to $B_{S'}$ for any $S'\in \binom{[n]}{n/2}, S'\neq S$.

The size of this pierceable family is $2^{\Omega(n)}$. It remains to notice that any intersection of  $p(n)$ sets from the family is realizable as a conjunction of $p(n)$ monotone disjunctions, and thus has at most $p(n)$ maxterms.

        \item[c)] The family of symmetric functions is $(n/2, n+1)$-rich and thus requires $\Omega(n)$ mistakes. Indeed, let $B_i$ be the set of strings with the number of 1s distinct from $i$. The family is  $B_0, \ldots, B_n$ is pierceable as for any $i = 0, \ldots, n$, the string $1^i 0^{n-i}$ belongs to all these sets except $B_i$. In turn, the intersection of any $n/2$ of these sets is the preimage of 1 of a symmetric function.

        \item[d)] For any $i, j\in\{1, \ldots, n\}, i\neq j$, consider the set $B_{i,j}$ of strings who which the disjunction $x_i\lor x_j$ is true. The family of sets $\{B_{i,j} \mid i, j\in\{1, \ldots, n\}, i\neq j\}$ is pierceable because for any  $i, j\in\{1, \ldots, n\}, i\neq j$, one can consider the string that has 1s everywhere except of the $i$-th and the $j$-th position -- that string belongs to all sets from the family except $B_{i,j}$. There are $\binom{n}{2}$ sets in the family. The intersection of any of them is realizable by a 2CNF. Hence, 2CNFs is a $(\binom{n}{2}/2, \binom{n}{2})$-rich family, which gives the $\Omega(n^2)$ bound on the number of mistakes.

        \item[e)] This is the only case when we require an argument not based on Proposition \ref{prop_generallb}. For $x,y\in\{0, 1\}^n$, denote:
        \[\langle x, y\rangle = x_1 y_1 + \ldots + x_n y_n \pmod{2}.\]
    For any $y\in\{0, 1\}^n$, our family contains a language of words $x\in\{0, 1\}^n$ with $\langle x, y\rangle = 1$.
    
        The adversary selects a uniformly random vector $r = r_1r_2\ldots r_n\in\{0, 1\}^n$. It enumerates all $x\in\{0, 1\}^n$ with $\langle x, r\rangle = 1$; this will constitute a valid play in the language generation game for our family. More specifically, for $k = 1, \ldots, n$, let $S_k(r)$ denote the set of strings $x\in\{0, 1\}^n$ such that $\langle x, r \rangle = 1$ and $x$ has 0s in the last $n - k$ bits. Observe that $S_1(r) \subseteq S_2(r) \subseteq \ldots \subseteq S_n(r)$, with $S_n(r)$ a language from our family. The adversary first enumerates all strings from $S_1(r)$, then the remaining strings from $S_2(r)$, then the remaining strings from $S_3(r)$, and so on -- at each step, in lexicographic order. After $S_k(r)$ is enumerated, the language generation algorithm is required to output a string $x'\notin S_k(r)$. If the last $n- k$ bits of $x'$ are 0s, then $\langle x', r\rangle = 0$ and the algorithm makes a mistake. 
      
        Now, assume that $x'$ has at least one 1 among the last $n - k$ bits.  Any fixation of $r_1, \ldots, r_k$ fixes the set $S_k(r)$, and thus the enumeration of the adversary in the first $k$ steps and the output  $x'$ of the algorithm after that. However, the distribution of $r_{k+1},\ldots, r_n$ is still uniform, and the probability that $\langle x', r\rangle$ will be equal to 1 is just $1/2$. Hence, on average, any language generation algorithm will make at least $(n-1)/2$ mistakes against this randomized strategy of the adversary.

    \end{itemize}
\end{proof}

\section{Conclusion}

The literature on language generation hides an implicit debate on whether this setting is an adequate mathematical reflection of the recent impressive practical success of LLMs. Some initial critiques of the initial positive result of Kleinberg and Mullainathan~\cite{kleinberg2024language}   -- like dependence of the mistake bound on the order of enumeration or the fact that only a small portion of the language is generated -- were addressed in a subsequent work~\cite{11369227,pmlr-v291-charikar25a}. However, Arenas et al.~\cite{arenas2025language} have observed that even for simple families, like families of regular languages with a few states, it might take an unrealistically long time for generation to become error-free, making this measure of success not adequate in the practical context. Kleinberg, Peale, and Reingold~\cite{kleinberg2026mistake} have partially addressed this by switching the focus from the moment of the last mistake to the total number of mistakes, giving an algorithm with a rather strong upper bound on the latter. 

However, the suggestion of Kleinberg, Peale, and Reingold does take into account the computational complexity of generation.  Our paper fills this gap and shows that unfortunately, even for rather simple families, generation with a polynomial number of mistakes is impossible in polynomial time (unless one-way functions exist). Curiously, we obtain this for a family of polynomial size (disjunctions of 3 literals) where one can go through all languages of the family in polynomial time and for which the Kleiberg-Peale-Reingold algorithm gives a logarithmic mistake bound (but works in exponential time). 

On the positive side, we show that a robust class of monotone functions -- those that have polynomially many maxterms -- does admit polynomial-time language generation with a polynomial number of mistakes. It leaves the question of how far one can extend this result for monotone functions?  Say, it is not apparent how to disprove the polynomial-time MBLG property for polynomial-size monotone circuits. Our current proof technique requires circuits to simulate the computation of cryptographic primitives, for which negations seem crucial.


\end{proof}




\end{document}